\documentclass[sigconf]{acmart}

\setcopyright{acmcopyright}
\copyrightyear{2023}
\acmYear{2023}
\acmDOI{}
\acmConference[]{}{}{}
\acmBooktitle{}
\acmPrice{}
\acmISBN{}

\usepackage{amsthm,amsmath}
\usepackage{tikz}
\usepackage{xfrac}

\usepackage{mathtools}
\usepackage{enumerate}

\newcommand{\bN} { {\mathbb{N}}}

\newcommand{\bC} { {\mathbb{C}}}

\newcommand{\bR} { {\mathbb{R}}}

\def\res{\operatorname{res}}

\newtheorem{thm}{Theorem}[section] \newtheorem{cor}[thm]{Corollary}
\newtheorem{lem}[thm]{Lemma}

\newtheorem{prop}[thm]{Proposition}

\newtheorem{exam}[thm]{Example}
\newtheorem{defn}[thm]{Definition}
\newtheorem{notation}[thm]{Notation}
\newtheorem{remark}[thm]{Remark}
\def\qed{\quad\rule{1ex}{1ex}}

\makeatletter
\newcommand{\myitem}[1]{%
\item[(#1)]\protected@edef\@currentlabel{#1}%
}
\makeatother

\def\eatspace#1{#1}
\def\step#1#2{\par\kern1pt\hangindent#2em\hangafter=1\noindent\rlap{\small#1}\kern#2em\relax\eatspace}

\def\<#1>{\langle#1\rangle}
\def\ord{\operatorname{ord}}

\def\lc{\operatorname{lc}}

\def\S{\mathcal{S}}

\newcommand{\cK}{\mathcal{K}}
\newcommand{\bZ}{\mathbb{Z}}
\newcommand{\cS}{\mathcal{S}}
\newcommand{\cN}{\mathcal{N}}
\newcommand{\cB}{\mathcal{B}}
\newcommand{\DR}{{k\langle D \rangle}}

\newcommand{\ind}{{\rm ind}}
\newcommand{\sind}{{\rm Sind}}
\newcommand{\den}{{\rm den}}
\newcommand{\num}{{\rm num}}
\def\ord{\operatorname{ord}}

\begin{document}
\fancyhead{}
\title{Stability Problems on D-finite Functions}
\thanks{{ S.\ Chen was partially supported by the NSFC
grants (No. \ 12271511 and No. \ 11688101), CAS Project for Young Scientists in Basic Research (Grant No.\ YSBR-034),
the Fund of the Youth Innovation Promotion Association (Grant No.\ Y2022001), CAS,
and the National Key Research and Development Project 2020YFA0712300. R. \ Feng, Z. \ Guo, W. \ Lu were partially supported by
the NSFC grants 11771433, 11688101, Beijing Natural Science Foundation under Grant
Z190004 and National Key Research and Development Project 2020YFA0712300.}}

\author{Shaoshi Chen, Ruyong Feng, Zewang Guo and Wei Lu}
%\orcid{1234-5678-9012}
%\author{G.K.M. Tobin}
%\email{webmaster@marysville-ohio.com}
\affiliation{%
  \institution{KLMM, Academy of Mathematics and Systems Science\\ University of Chinese Academy of Sciences}
  \country{Chinese Academy of Sciences, Beijing 100190, China}
%  \streetaddress{P.O. Box 1212}
}
\email{{schen, ryfeng, guozewang}@amss.ac.cn, luwei18@mails.ucas.ac.cn}
%
%\author[S.\ Chen, L.\ Du, M.\ Kauers]{Shaoshi Chen$^{a, b}$, Lixin Du$^{c}$, Manuel Kauers$^c$}
%
%\affiliation{%
%  \institution{$^a$KLMM, Academy of Mathematics and Systems Science, Chinese Academy of Sciences, Beijing 100190, China}
%  \institution{$^b$School of Mathematical Sciences, University of Chinese Academy of Sciences, Beijing 100049, China}
%  \institution{$^c$Institute for Algebra, Johannes Kepler University, Linz, A4040,  Austria}
%  \state{}
%  \postcode{}
%  \country{}
%}
%\email{schen@amss.ac.cn, lixin.du@jku.at, manuel.kauers@jku.at}
%\email{manuel.kauers@jku.at}

\begin{abstract}
 This paper continues the studies of symbolic integration by
 focusing on the stability problems on D-finite functions. We introduce the notion of stability index in order to
 investigate the order growth of the differential operators satisfied by iterated integrals of D-finite functions and
 determine bounds and exact formula for stability indices of several special classes of differential operators.
 With the basic properties of stability index, we completely solve the stability problem on general hyperexponential functions.

\end{abstract}
\begin{CCSXML}
	<ccs2012>
	<concept>
	<concept_id>10010147.10010148.10010149.10010150</concept_id>
	<concept_desc>Computing methodologies~Algebraic algorithms</concept_desc>
	<concept_significance>500</concept_significance>
	</concept>
	</ccs2012>
\end{CCSXML}

\ccsdesc[500]{Computing methodologies~Algebraic algorithms}

\keywords{D-finite functions; stable Sets; symbolic integration}
\maketitle

\section{Introduction}

New functions were introduced when certain integrals could not be evaluated in term of a specific
class of functions. We need to introduce the logarithmic function $\log(x)$ since $1/x$ has no
primitive in the field of rational functions. The special functions, such as $\int \exp(x^2) \, dx, \int 1/\log(x) \, dx$ etc., are introduced likewise since they are not elementary. We can continue to study the integrals of the newly introduced functions,
e.g., the indefinite integral of $\log(x)$ is $x\log(x)-x$.  A natural question is whether we need keep introducing new functions in order to integrate a given function iteratively.

Motivated by irrationality proofs in number theory and the above question, Chen initialized the dynamical aspect of symbolic integration  by studying
some stability problems in differential fields in~\cite{chen2022ISSAC}. For a given special function $f$ in certain differential ring $(R, D)$, the stability problem is
deciding whether there exists a sequence $\{g_i\}_{i\geq 0}$ in $R$ such that $f = D^i(g_i)$ for all $i\in \bN$.
Such stable special functions appear in any integro-differential algebra, which is an algebra equipped with derivation and integration operators.  The theory of integro-differential algebras has been developed in~\cite{Guo2014}. Two basic examples of integro-differential algebras are the ring of polynomials $\bC[x]$ and the ring $C^\infty(\bR)$ of smooth functions on $\bR$. From the algorithmic point of view, the algebra  $C^\infty(\bR)$ is too big and we are more interested in some subalgebra whose elements are computable, such as the subalgebras $\bC[x, \log(x)]$ and $\bC[x, \exp(x)]$. To find more rich examples, it is necessary to classify certain class of functions that are iteratively integrable in this class.

Focusing on elementary functions, Chen in~\cite{chen2022ISSAC} had solved  the stability problems on three families of special functions including rational functions, logarithmic functions, and exponential functions. The main goal of this paper is to continue the work~\cite{chen2022ISSAC} by focusing on D-finite functions using the language of differential operators.

The notion of D-finite power series was first introduced by Stanley~\cite{Stanley1980} in 1980 in the univariate case and later studied by Lipshitz in the multivariate case~\cite{Lipshitz1989}. D-finite functions are
also called holonomic functions that had played a significant role in Zeilberger's method of creative telescoping~\cite{Zeilberger1990, Wilf1992}.
These series, like algebraic numbers, can be algorithmically manipulated via its defining linear differential equations~\cite{AbramovLeLi2005, Salvy2019}. For a comprehensive
introduction to theory and algorithms for D-finite functions, one can see the coming monograph by Kauers~\cite{DfiniteBook}.
The indefinite integration problem on D-finite functions
was studied by Abramov and van Hoeij in~\cite{AbramovHoeij1997}.
The main goal of this paper  is to investigate Abramov-van Hoeij's algorithm iteratively.

The remainder of this paper is organized as follows. We recall some basic terminologies in differential algebra and then define
the notion of principal integrals related to iterated integration in Section~\ref{sec:principalintegrals}. We will focus on studying the order of the differential operators satisfied
by principal integrals in Section~\ref{sec:orderoftheintegrals}, which inspires us to introduce the notion of stability index.  In Section~\ref{sec:stability}, we will
determine the stability indices of two special classes of differential operators including Katz's operators and first-order operators.
With the help of stability index, we will completely solve the stability problem for general hyperexponential functions in Section~\ref{sec:hyperexp}
which generalizes the results in~\cite{chen2022ISSAC}. As an application, we present an exact formula for the stability index of a special class of rational functions in Section~\ref{sec:rationalfunctions}.

\section{Principal integrals}\label{sec:principalintegrals}
Throughout this paper, let $C$ be any algebraically closed field of characteristic zero and $k=C(x)$ be
the differential field with the usual derivation $'$ such that $x'=1$ and $c'=0$ for all $c\in C$.
Let $\cK$ be a differential closure of $C(x)$ (see Kolchin's book~\cite{kolchin} for the existence of such closures).
Any system of algebraic differential equations with coefficients from $k$ has solutions in $\cK$ if it has any solution in some differential extension of $k$.
Let $k\langle D \rangle$ be the ring of linear differential operators with coefficients in $k$ in which we have
\[D \cdot f = f\cdot D + f', \quad \text{for any $f\in k$}.\]
The ring $k\langle D \rangle$ is a left Euclidean domain in which any left ideal is principal~\cite{Ore1933}.
For an $m\in \bZ$, $\bZ_{\leq m}$ (resp., we let $\bZ_{\geq m}$) stand for the set of all integers not greater than $m$ (resp., not less than $m$).

Let $L$ be a nonzero operator in $k\langle D \rangle$. Write
\[
    L= \sum_{i=0}^{n} a_i D^i, \quad \text{where $a_i  \in k$ and $a_n\neq 0$}.
\]
The integer $n$ is called the order of $L$, denoted by $\ord(L)$, and $a_n$ is called the leading coefficient
of $L$, denoted by $\lc(L)$. If $a_n=1$, then $L$ is called a monic operator.
The adjoint operator of $L$, denoted by $L^*$, is defined to be
\[
    L^* := \sum_{i=0}^n (-D)^i a_i.
\]

The field $\cK$ can be equipped with a left $\DR$-module structure with the action $L(f) = \sum_{i=0}^{n} a_i f^{(i)}$ for any $L=\sum_{i=0}^{n} a_i D^i\in \DR$ and $f\in \cK$.
The annihilating ideal
\[\text{Ann}_\DR(f) := \{P\in \DR \mid P(f)=0\}\] is a left ideal of $\DR$. Since $\DR$ is a left Euclidean domain,
we have  $\text{Ann}_\DR(f) = \langle L \rangle$ for some monic operator $L\in \DR$.
\begin{defn}
An element $f\in \cK$ is said to be~\emph{D-finite} over $k$ if $\text{Ann}_\DR(f) =  \langle L \rangle$ for some nonzero and monic operator $L\in \DR$.
We call $L$ the defining operator for $f$, denoted by $L_f$, whose order is called the order of $f$, denoted by $\ord(f)$.
For convenience, we set $\ord(f)=\infty$ if $f$ is not D-finite.
\end{defn}

Lipshitz~\cite{Lipshitz1989} showed that the set of all D-finite elements over $k$ forms a subalgebra of $\cK$, which is also closed under taking
derivatives and integrals. We will investigate how the order changes when integrating a D-finite function.

\begin{defn}
For $f\in \cK$ and $i\in \bZ_{>0}$, $g\in \cK$ is called an $i$-th primitive or an $i$-th integral of $f$ if $f=D^i(g)$. When $i=1$, we simply call $g$ an integral of $f$.
\end{defn}

Suppose $f$ is a D-finite function and $g_i$ is an $i$-th integral of $f$. Then $g_i$ is also D-finite and $\ord(g_i)\leq \ord(f)+i$. Set
\begin{equation*}
\cN_i(f)=\min \{\ord(g) \mid \mbox{$g$ is an $i$-th integral of $f$}\}.
\end{equation*}
The following example shows that different i-th integrals of a given function may have different orders.

{
\begin{exam}
\label{exam:principalintegrals}
  Let $f=\exp(x)$. Then $f$ itself is a second integral of $f$ and $\ord(f)=1$. Since every second integral of $f$ must be of order not less than 1, $\cN_2(f)=1$. On the other hand, there exist second integrals of $f$ that are of order greater than 1. In fact, each second integral of $f$ is of the form $\exp(x)+\alpha x+\beta$ with $\alpha,\beta\in C$. Assume that not all $\alpha,\beta$ are zero. We claim that $\ord(\exp(x)+\alpha x+\beta)=2$. First, as $\exp(x)+\alpha x+\beta$ is not hyperexponential, $\exp(x)+\alpha x+\beta$ can not be annihilated by any first-oder operator. Second, if $\alpha\neq 0$ then $\exp(x)+\alpha x+\beta$ is annihilated by
  $$
     \left(x+\frac{\beta}{\alpha}+1\right)D^2-\left(x+\frac{\beta}{\alpha}\right)D-1,
  $$
  and if $\alpha=0$ and $\beta\neq 0$ then $\exp(x)+\beta$ is annihilated by $D^2-D$. So $\ord(\exp(x)+\alpha x+\beta)=2$ and the claim is shown.
\end{exam}
The example above motivates the following notion of principal integrals.
}
\begin{defn}
\label{def:principalintegraloffunctions}
The $i$-th integrals of order $\cN_i(f)$ are called the principal $i$-th integrals of $f$.
\end{defn}
{It can happen that all integrals are principal as shown in the following example.
\begin{exam}
Let $f=\frac{1}{x}$. Then each principal second integral of $f$ is of the form $x\ln(x)+\alpha x+\beta$ with $\alpha,\beta\in C$. It is easy to check that $x\ln(x)+\alpha x+\beta$ is annihilated by
$$
   (x-\beta)D^2-D+\frac{1}{x}.
$$
A similar argument as in Example~\ref{exam:principalintegrals} implies that $\ord(x\ln(x)+\alpha x+\beta)=2$ for all $\alpha,\beta\in C$.
\end{exam}
}
It was proved in~\cite{AbramovHoeij1997,AbramovHoeij1999} that $\cN_1(f)=\ord(f)$ if and only if $L_f^*(y)=1$ has a solution in $k$. If $L_f^*(y)=1$ has no solution in $k$ then $\cN_1(f)=\ord(f)+1$ and moreover $L_g=L_fD$ for any integrals $g$ of $f$. Suppose that $L_f^*(y)=1$ has a solution $l$ in $k$. Then $(1-lL_f)^*(1)=0$ and so there is a unique $H_l\in \DR$ of order $n-1$ such that
\begin{equation*}
\label{eqn:formula1}
    1-lL_f=DH_l
\end{equation*}
where $n=\ord(f)$.
Note that $-l$ is the leading coefficient of $H_l$.
Set $g=H_l(f)$. Then $g$ is an integral of $f$ and
\begin{equation*}
\label{eqn:formula2}
    L_g =\frac{1}{l}(1-H_lD).
\end{equation*}
For a nonzero operator $P\in \DR$ and $q\in C[x]$, define
$$
   \delta(P,q)=\begin{cases} 1 & \mbox{$P^*(y)=q$ has a solution in $C(x)$}\\
     0 & \mbox{otherwise}
   \end{cases}.
$$
When $q=1$, we abbreviate $\delta(P,q)$ to $\delta(P)$. The above discussions motivate the definition below.
\begin{defn}
\label{defn:principalintegralofoperator}
Suppose that $L$ is a nonzero monic operator in $\DR$. A principal integral of $L$ is defined to be
$$
\begin{cases}
    LD & \delta(L)=0\\
    \frac{1}{l}(1-H_lD) & \delta(L)=1
\end{cases}
$$
where $l\in k$ is a solution of $L^*(y)=1$ and $H_l$ is the unique operator in $\DR$ of order $\ord(L)-1$ such that $lL+DH_l=1$. We also call $\frac{1}{l}(1-H_lD)$ the principal integral of $L$ with respect to $l$. An $i$-th principal integral of $L$ is defined iteratively.
\end{defn}
\begin{remark}\label{rem:shouldprovewelldefined}
  {Notice that in the case $\delta(L)=1$ principal integrals of $L$ depend on the rational solutions of $L^*(y)=1$ and thus they are generally not unique. For example, let $L=D$. The adjoint $L^*$ of $L$ is $-D$. For every $c\in C$, $-x+c$ is a rational solution of $L^*(y)=1$. One sees that $x-c$ satisfies that $(-x+c)D+D(x-c)=1$.
  Since
  $$
     \frac{1}{-x+c}(1-(x-c)D)=D-\frac{1}{x-c},
  $$
  by definition, $D-\frac{1}{x-c}$ is a principal integral of $L$ for any $c\in C$. Therefore the principal integrals of $L$ are not unique. However they have the same order.}
\end{remark}
% \begin{remark}\label{rem:principalintegrals}
% Suppose that $f\in \cK$ and $g$ is a principal integral of $f$. Then from Definition~\ref{defn:principalintegralofoperator} one sees that $L_g$ is a principal integral of $L_f$. On the other hand, suppose that $\tilde{L}$ is a principal integral of $L$ and $f\in \cK$ satisfies that $L=L_f$. Then there is a principal integral $g$ of $f$ such that $\tilde{L}=L_g$.
% \end{remark}
\begin{lem}
\label{lm:principalintegrals}
{Suppose that $f\in \cK$ and $L_i$ is an $i$-th principal integral of $L_f$. Then there is an $i$-th integral of $f$, say $g_i$, such that $L_i=L_{g_i}$.}
\end{lem}
\begin{proof}
{We shall show the lemma by induction on $i$. The case $i=1$ follows from Definition~\ref{defn:principalintegralofoperator} and the discussions before Definition~\ref{defn:principalintegralofoperator}. Suppose that $i>1$ and the assertion holds for the case $i-1$. Let $L_{i-1}$ be an $(i-1)$-th principal integral of $L_f$ and $L_i$ is a principal integral of $L_{i-1}$. By induction hypothesis, there is an $(i-1)$-th integral $g_{i-1}$ of $f$ such that $L_{i-1}=L_{g_{i-1}}$. Since $L_i$ is a principal integral of $L_{i-1}=L_{g_{i-1}}$, by induction hypothesis again, there is an integral $g_i$ of $g_{i-1}$ such that $L_i=L_{g_i}$. The lemma then follows from the fact that $g_i$ is an $i$-th integral of $f$.}
\end{proof}
{In next section, we shall show that the $g_i$ in Lemma~\ref{lm:principalintegrals} is actually an $i$-th principal integral of $f$. Furthermore, we shall show that $\ord(L_i)=\cN_i(f)$, and Definition~\ref{def:principalintegraloffunctions} and Definition~\ref{defn:principalintegralofoperator} are consistent in some sense.
 }

% \begin{remark}
% \label{rem:orders}
% Suppose $P$ is a principal integral of $L$.
% From the definition, one sees that $\ord(P)=\ord(L)-\delta(L)+1$.
% \end{remark}
%  Remark~\ref{rem:orders} indicates that all principal integrals of $L$ have the same order. In next section, we shall show that all principal $i$-th integrals of $L$ also have the same order.

\section{Orders of principal integrals}
\label{sec:orderoftheintegrals}
Throughout this section, $L$ is always a monic operator in $\DR$.
 Suppose that $L_i$ is a principal $i$-th integral of $L$ for any positive integer $i$.
We are going to prove that $\ord(L_{i+1})=\ord(L_i)$ if and only if there is a polynomial $p\in C[x]$ of degree $i$ such that $L_i^*(y)=p$ has a solution in $k$. This generalizes the results in~\cite{AbramovHoeij1997,AbramovHoeij1999} for the case when $i=1$. For a rational function $l\in k$, $l$ is said to be rationally integrable if $l=D(h)$ for some $h\in k$.
\begin{lem}
\label{lm:rationallyintegrable}
Suppose that $P=1+M D$ and $l$ is a rational solution of $P^*(y)=p$, where $M\in \DR$ and $p\in C[x]$. Then $l$ is rationally integrable.
\end{lem}
\begin{proof}
Note that $P^*=1-DM^*$. Since $l$ is a solution of $P^*(y)
 = p$ in $k$, we have $l-D(M^*(l)) = p$.
Let $q\in C[x]$ be such that $p = q'$. Then $l=D(q + M^*(l))$. The lemma then follows from the fact that $q+M^*(l)\in k$.
\end{proof}

\begin{prop}
\label{prop:expressionofLDn}
Suppose that $L_0=L$ and $L_{i+1}$ is a principal integral of $L_{i}$ for any $i\geq 0$. Set $I_{-1}=1$ and for each $i\geq 0$ set
$$
    I_i=\begin{cases}
        1 & \delta(L_i)=0\\
       D+\frac{D(l_i)}{l_i} & \delta(L_i)=1
    \end{cases}
$$
where $l_i\in k$ is a solution of $L_i^*(y)=1$ such that $L_{i+1}$ is the principal integral of $L_i$ with respect to $l_i$.
Then for each $n\geq 0$,
\begin{equation}
\label{eqn:expressionLDn}
    LD^n = \left(\prod_{i=-1}^{n-1} I_i \right)L_n.
\end{equation}
\end{prop}
\begin{proof}
We shall prove the proposition by induction on $n$. The case $n=0$ is clear. Assume that $n>0$ and the assertion holds for $n-1$. By induction hypothesis,
$$
 LD^{n-1}=\left(\prod_{i=-1}^{n-2}I_i\right)L_{n-1}.
$$
If $\delta(L_{n-1})=0$ then $L_n=L_{n-1}D$ and $I_{n-1}=1$. Thus
$$
  LD^n=LD^{n-1}D=\left(\prod_{i=-1}^{n-2}I_i\right)L_{n-1}D=\left(\prod_{i=-1}^{n-1}I_i\right)L_{n}.
$$
If $\delta(L_{n-1})=1$ then $l_{n-1}L_{n-1}+DH_{l_{n-1}}=1$ and $L_n=\frac{1}{l_{n-1}}(1-H_{l_{n-1}}D)$. These imply that
\begin{align*}
  L_{n-1}D&=\frac{1}{l_{n-1}}(1-DH_{l_{n-1}})D=\frac{1}{l_{n-1}}D(1-H_{l_{n-1}}D)\\
  & =\frac{1}{l_{n-1}}D l_{n-1} L_n =\left(D+\frac{D(l_{n-1})}{l_{n-1}}\right)L_n=I_{n-1}L_n.
\end{align*}
Therefore
$$
  LD^n=LD^{n-1}D=\left(\prod_{i=-1}^{n-2}I_i\right)L_{n-1}D=\left(\prod_{i=-1}^{n-1}I_i\right)L_n.
$$
The assertion holds for all nonnegative integers $n$.
\end{proof}
% \begin{lem}
% \label{lm:principalintegrals}
% Suppose that $P_1$ and $P_2$ are principal integrals of $L$. Then $P_1^*(y)=1$ has a solution in $k$ if and only if $P_2^*(y)=1$ has too.
% \end{lem}
% \begin{proof}
% If $L^*(y)=1$ has no solutions in $k$ then $P_1=P_2=LD$ and there is nothing to prove. Suppose that $L^*(y)=1$ has at least one solution in $k$. Then $P_i=\frac{1}{l_i}(H_{l_i}D-1)$ for $i=1,2$, where $l_1,l_2\in k$ are solutions of $L^*(y)=1$ and $H_{l_i}$ is given as in Definition~\ref{defn:principalintegralofoperator}.
% Now suppose that $h_1\in k$ satisfies that $P_1^*(h_1)=1$. Then $h_1/l_1$ is a solution of $(1-H_{l_1}D)^*(y)=1$ in $k$. Due to Lemma~\ref{lm:rationallyintegrable}, there is $s_1\in k$ such that $h_1/l_1=D(s_1)$. From $l_iL+DH_{l_i}=1$, one sees that
% $$
%     l_i L D=D(1-H_{l_i}D).
% $$
% Taking adjoint on both sides yields that
% $$
%     DL^* l_i=(1-H_{l_i}D)^*D.
% $$
% Now set $h_2=l_2D(l_1s_1/l_2)$. Then $h_2\in k$ and
% \begin{align*}
% P_2^*(h_2)&=(H_{l_2}D-1)^*\frac{1}{l_2}l_2D(l_1s_1/l_2)=-DL^*l_2(l_1s_1/l_2)\\
% &=DL^*l_1s_1=(1-H_{l_1}D)^*D(s_1)\\
% &=(1-H_{l_1}D)^*(h_1/l_1)=P_1^*(h_1)=1.
% \end{align*}
% The case that $P_2^*(y)=1$ has a solution in $k$ can be proved similarly.
% \end{proof}
\begin{lem}
\label{lm:rationalsolutions}
Suppose that $L_{n}$ is a principal integral of $L_{n-1}$. Then for every $p\in C[x]$, $\delta(L_n,p)=\delta(L_{n-1},q)$ for some $q\in C[x]$ with $p=q'$. In other words, $L_n^*(y)=p$ has a solution in $k$ if and only if $L_{n-1}^*(y)=q$ has a solution in $k$.
\end{lem}
\begin{proof}
Assume that $\delta(L_{n-1})=0$. Then $L_{n}=L_{n-1}D$ and thus $L_{n}^*=-DL_{n-1}^*$. One can verify that for each $h\in k$, $L_{n}^*(h)=p$ if and only if $L_{n-1}^*(-h)=q$ for some $q\in C[x]$ with $p=q'$. Now suppose that $\delta(L_{n-1})=1$ and $L_{n}$ is the principal integral of $L_{n-1}$ with respect to $l_{n-1}$. Then $L_{n}=1/l_{n-1}(1-H_{l_{n-1}}D)$, where $H_{l_{n-1}}\in \DR$ satisfies that $l_{n-1}L_{n-1}+DH_{l_{n-1}}=1$. Suppose that $h\in k$ is a solution of $L_{n}^*(y)=p$. Then $h/l_{n-1}$ is a solution of $(1-H_{l_{n-1}}D)^*(y)=p$ in $k$. Due to Lemma~\ref{lm:rationallyintegrable}, there is a $g\in k$ such that $D(g)=h/l_{n-1}$. Since $l_{n-1}L_{n-1}D=D-DH_{l_{n-1}}D=D(1-H_{l_{n-1}}D)$, one has that
\begin{align*}
   DL_{n-1}^*(l_{n-1}g)&=DL_{n-1}^*l_{n-1}(g)=(1-H_{l_{n-1}}D)^*D(g)\\
   &=(1-H_{l_{n-1}}D)^*(h/l_{n-1})=p.
\end{align*}
Let $\tilde{q}\in C[x]$ be such that $D(\tilde{q})=p$. Then $D(L_{n-1}^*(l_{n-1}g)-\tilde{q})=0$ and thus $L_{n-1}^*(l_{n-1}g)-\tilde{q}=c\in C$. Set $q=\tilde{q}+c$. Then $L_{n-1}^*(l_{n-1}g)=q$. Conversely, assume that $h$ is a solution of $L_{n-1}^*(y)=q$ for some $q\in C[x]$ with $p=q'$. Then $DL_{n-1}^*(h)=p$. Set $\tilde{h}=l_{n-1}D(h/l_{n-1})$. One then has that
\begin{align*}
   L_{n}^*(\tilde{h})&=(1-H_{l_{n-1}}D)^*(\tilde{h}/l_{n-1})=(1-H_{l_{n-1}}D)^*D(h/l_{n-1})\\
   &=(1+DH_{l_{n-1}}^*)D\frac{1}{l_{n-1}}(h)=D(1+H_{l_{n-1}}^*D)\frac{1}{l_{n-1}}(h)\\
   &=D\left(\frac{1}{l_{n-1}}(1-DH_{l_{n-1}})\right)^*(h)=DL_{n-1}^*(h)=p.
\end{align*}
In other words, $L_{n}^*(y)=p$ has a solution in $k$.
\end{proof}

\begin{prop}
\label{prop:ratioalsolutions}
Let $L_i$ be as in Proposition~\ref{prop:expressionofLDn}.
Then
$$
\delta(L_n,p)=\delta(LD^n,p)
$$
for any $p\in C[x]\setminus\{0\}$ and nonnegative integer $n$.
\end{prop}
\begin{proof}
We shall prove the assertion by induction on $n$. The case $n=0$ is clear, because $L_0=L=LD^0$ in this case. Suppose that $n>0$ and the assertion holds for the case $n-1$.
Consider the case $n$. It suffices to show that $\delta(L_n,p)=1$ if and only if $\delta(LD^n,p)=1$.
Assume that $\delta(L_n,p)=1$. Lemma~\ref{lm:rationalsolutions} implies that $\delta(L_{n-1},q)=1$ for some $q\in C[x]$ with $p=q'$. By induction hypothesis, $\delta(LD^{n-1},q)=1$, i.e. there is an $r\in k$ such that $(LD^{n-1})^*(r)=q$. Since $(LD^n)^*=-D(LD^{n-1})^*$, one sees that
$(LD^n)^*(-r)=D((LD^{n-1})^*(r))=q'=p$, i.e. $\delta(LD^n,p)=1$. Conversely, assume that $\delta(LD^n,p)=1$. Then $\delta(LD^{n-1},\tilde{q})=1$ for some $\tilde{q}\in C[x]$ with $D(\tilde{q})=p$.  By induction hypothesis again, $\delta(L_{n-1},\tilde{q})=1$. By Lemma~\ref{lm:rationalsolutions} again, $\delta(L_n,p)=1$. Hence the assertion holds for the case $n$.
\end{proof}

%\begin{cor}
%\label{cor:orderofintegrals}
%All $i$-th principal integrals of $L$ have the same order.
%\end{cor}
%\begin{proof}
%We shall prove the corollary by induction on $i$. The case $i=0$ is clear. Suppose $i>0$ and the assertion holds for the case $i-1$. Consider the case $i$.
%Suppose that $P_1$ and $P_2$ are two $i$-th principal integrals of $L$. Suppose further that $P_1, P_2$ are principal integrals of $Q_1, Q_2$ respectively. Then $Q_1$ and $Q_2$ are $i-1$-th principal integrals of $L$.
%By induction hypothesis, $\ord(Q_1)=\ord(Q_2)$. Due to Proposition {~\ref{prop:ratioalsolutions}} and Remark~\ref{rem:orders},
%\begin{align*}
% \ord(P_1)&=\ord(Q_1)-\delta(Q_1)+1=\ord(Q_2)-\delta(LD^{i-1})+1\\
%&=\ord(Q_2)-\delta(Q_2)+1=\ord(P_2).
%\end{align*}
%The second and third equalities hold because $\delta(Q_1)=\delta(LD^{i-1})=\delta(Q_2)$ by Proposition~\ref{prop:ratioalsolutions} with $p=1$.
%\end{proof}

\begin{cor}
\label{cor:criterion}
Let $L_i$ be as in Proposition~\ref{prop:expressionofLDn}. Then $\ord(L_{i+1})=\ord(L_i)$ if and only if $\delta(L,q)=1$ for some $q\in C[x]$ of degree $i$.
\end{cor}
\begin{proof}
From Definition~\ref{defn:principalintegralofoperator}, $\ord(L_{i+1})=\ord(L_i)$ if and only if $\delta(L_i)=\delta(L_i,1)=1$. Note that $\delta(L_i)=\delta(L_i,c)$ for any nonzero $c\in C$. Using Lemma~\ref{lm:rationalsolutions} repeatedly, there is a $q\in C[x]$ of degree $i$ such that
$$
   \delta(L,q)=\delta(L_1, q')=\dots=\delta(L_i,i!\lc(q))=\delta(L_i).
$$
In other words, for such $q$, $\delta(L_i)=1$ if and only if $\delta(L,q)=1$. Hence $\ord(L_{i+1})=\ord(L_i)$ if and only if $\delta(L,q)=1$ for some $q\in C[x]$ of degree $i$.
\end{proof}
\begin{lem}
\label{lm:orderofprincipalintegras}
 {Let $f\in \cK$ and $L_i$ as in Proposition~\ref{prop:expressionofLDn} with $L=L_f$. Then $\ord(L_i)=\cN_i(f)$.}
\end{lem}
\begin{proof}
 {Set
$$
   S_i(f)=\{ 0\leq j \leq i-1 \mid \delta(L_f D^j)=1\}.
$$
Then by Proposition~\ref{prop:ratioalsolutions}, $\ord(L_i)=i+\ord(f)-|S_i(f)|$. Therefore,
it suffices to show that $|S_i(f)|= i-\cN_i(f)+\ord(f)$.
We shall prove this by induction on $i$. The case $i=1$ follows from the fact that $\cN_1(f)=\ord(f)$ if and only if $\delta(L_f)=1$. Suppose that $i>1$ and the assertion holds true for the case $i-1$. Let $g$ be an $i$-th principal integral of $f$. Then any $(i-1)$-th integral $h$ of $D^{i-1}(g)$ is an $i$-th integral of $f$. Hence $\ord(h)\geq \ord(g)$ and so $g$ is an $(i-1)$-th principal integer of $D^{i-1}(g)$. In other words, $\cN_{i-1}(D^{i-1}(g))=\cN_i(f)$.
By induction hypothesis,
$$
   S_{i-1}(D^{i-1}(g))=i-1-\cN_i(f)+\ord(D^{i-1}(g)).
$$
Assume that $D^{i-1}(g)$ is a principal integral of $f$.
Set $\tilde{L}_1=L_{D^{i-1}(g)}$. Let $\tilde{L}_{j+1}$ be a principal integral of $\tilde{L}_j$ for $1\leq j \leq i-1$. As $\tilde{L}_1$ is a principal integral of $L_f$, $\tilde{L}_j$ is a $j$-th principal integral of $L_f$.
By Proposition~\ref{prop:ratioalsolutions} $\delta(\tilde{L}_1D^j)=1$ if and only if $\delta(\tilde{L}_{j+1})=1$, and so if and only if $\delta(L_f D^{j+1})=1$. Thus
  $
    \{s+1 \mid s\in S_{i-1}(D^{i-1}(g))\}\subset S_i(f).
 $ On the other hand, one has that $0\in S_i(f)$ if and only if $\ord(D^{i-1}(g))=\ord(f)$. Hence $|S_i(f)|$ is not less than
\begin{align*}
 &|\{s+1 \mid s\in S_{i-1}(D^{i-1}(g))\}|+1+\ord(f)-\ord(D^{i-1}(g))\\
 &\geq i-1-\cN_i(f)+\ord(D^{i-1}(g))+1+\ord(f)-\ord(D^{i-1}(g))\\
 &=i-\cN_i(f)+\ord(f).
\end{align*}
If $|S_i(f)|>i-\cN_i(f)+\ord(f)$ then $\ord(L_i)<\cN_i(f)$. By Lemma~\ref{lm:principalintegrals} there is an $i$-th integral $\tilde{g}$ of $f$ such that $L_{\tilde{g}}=L_i$. So $\ord(\tilde{g})<\cN_i(f)$, a contradiction. Hence $|S_i(f)|=i-\cN_i(f)+\ord(f)$.
It remains to show that $D^{i-1}(g)$ is a principal integral of $f$.
Assume on the contrary that $D^{i-1}(g)$ is not a principal integral of $f$. Then $L_{D^{i-1}(g)}=L_fD$ and $\cN_1(f)=\ord(f)$. These imply that $\delta(L_fD^{s+1})=1$ for all $s\in S_{i-1}(D^{i-1}(g))$ and $\delta(L_f)=1$. These imply that
\begin{align*}
  \ord(L_i)&\leq i+\ord(f)-|S_{i-1}(D^{i-1}(g))|-1\\
  &\leq \cN_i(f)-\ord(D^{i-1}(g))+\ord(f)=\cN_i(f)-1.
\end{align*}
Using a similar argument as before, one will obtain a contradiction. Thus $D^{i-1}(g)$ must be a principal integral of $f$.  }
\end{proof}
 {The following proposition shows that Definition~\ref{def:principalintegraloffunctions} and Definition~\ref{defn:principalintegralofoperator} are consistent.}
\begin{prop}
\label{prop:definitionsequivalence}
 {
Suppose that $f\in \cK$ and $L_i$ is an $i$-th principal integral of $L_f$.
\begin{enumerate}
\item
Suppose that $g$ is an $i$-th integral of $f$ annihilated by $L_i$. Then $g$ is an $i$-th principal integral of $f$.
\item
Suppose that $g$ is an $i$-th principal integral of $f$. Then $L_g$ is an $i$-th principal integral of $L_f$.
\end{enumerate}}
\end{prop}
\begin{proof}
 {(1). By Lemma~\ref{lm:principalintegrals}, $\ord(g)\leq \ord(L_i)=\cN_i(f)$. So $g$ must be an $i$-th principal integral of $f$.}

 {(2). We shall show the assertion by induction on $i$. The case $i=1$ follows from the discussion after Proposition 4 of \cite{AbramovHoeij1999}. Suppose that $i>1$ and the assertion holds true for the case $i-1$. The proof of Lemma~\ref{lm:principalintegrals} implies that $g$ is an $(i-1)$-th principal integral of $D^{i-1}(g)$ and $D^{i-1}(g)$ is a principal integral of $f$. By induction hypothesis, $L_g$ is an $(i-1)$-th principal integral of $L_{D^{i-1}(g)}$ and $L_{D^{i-1}(g)}$ is a principal integral of $L_f$. Hence $L_g$ is an $i$-th principal integral of $L_f$.}
\end{proof}
Given a nonzero $P\in \DR$, there exist a nonzero polynomial $\ind^P(s)\in C[s]$ and an integer $\sigma^P$ such that for any $s\in \bZ$,
$$
 P(x^s)=\ind^P(s)x^{s+\sigma^P}(1+c_1x^{-1}+c_2x^{-2}+\dots)
$$
where $c_i\in C$ (see~\cite[p. 102]{vanderput-singer}).
The polynomial $\ind^P$ is usually called the indicial polynomial of $P$ at $\infty$.
Remark that $Pd$ and $P$ have the same indicial polynomial and $\sigma^{Pd}=\sigma^P+\deg(d)$, where $d$ is a nonzero monic polynomial. In the following, $V_{\bZ_{\geq 0}}(\ind^P)$ stands for the set of nonnegative integer solutions of $\ind^P(s)=0$ and we agree  {that} $\max \emptyset =-1$.
\begin{thm}
\label{thm:bound}
Let $d$ be a nonzero monic polynomial of minimal degree such that $dL$ is an operator of polynomial coefficients. Set
\begin{equation}
\label{eqn:bound}
\cB(L)=\max\{0,\max V_{\bZ_{\geq 0}}(\ind^{L^*})+1+\sigma^{L^*}+\deg(d)\}.
\end{equation}
Then for all $i\geq \cB(L)$, there exists a polynomial $p$ of degree $i$ such that $\delta(L,p)=1$.
\end{thm}
\begin{proof}
It is clear that $\delta(dL,p)=\delta(L,p)$ for any nonzero polynomial $p$. So it suffices to consider $dL$. Note that $\ind^{(dL)^*}=\ind^{L^*}$ and $\sigma^{(dL)^*}=\sigma^{L^*}+\deg(d)$. Suppose that $i\geq \cB(L)$. Set $\ell=i-\sigma^{(dL)^*}$. Then $\ell\geq 0$ and $(dL)^*(x^\ell)$ is a polynomial. Furthermore,
\begin{align*}
   (dL)^*(x^\ell)&=\ind^P(\ell)x^{\ell+\sigma^{(dL)^*}}+\mbox{lower terms}\\
   &=\ind^P(\ell)x^i+\mbox{lower terms}.
\end{align*}
Since $\ell\geq \cB(L)-\sigma^{(dL)^*}\geq \max V_{\bZ_{\geq 0}}(\ind^P)+1$, $\ind^{(dL)^*}(\ell)\neq 0$ and thus $\deg((dL)^*(x^\ell))=i$. Set $p=(dL)^*(x^\ell)$. Then $p$ is a polynomial as required.
\end{proof}

In the case of D-finite power series, it has been proved in~\cite[Theorem 4.5]{chen2022ISSAC} that the order of $i$-th principal integrals of a given D-finite function
is uniformly bounded.  We now provide a more explicit order bound in terms of the information of indicial polynomials.

\begin{defn}
\label{def:stabilityindex}
The stability index of $L$, denoted by $\sind(L)$, is defined to be the minimal nonnegative integer $m$ satisfying that for each $i\geq m$ there exists a polynomial $p_i$ of degree $i$ such that $\delta(L,p_i)=1$. For an $f\in \cK$, the stability index of $f$ is defined to be the stability index of $L_f$, also denoted by $\sind(f)$.
\end{defn}

\begin{remark}
\label{rem:stabilityindex}
It is clear that $\sind(L)\leq \cB(L)$.
\end{remark}
\begin{exam}
\label{exam:notpolynomialsolution}
 {Let $L=D^3-xD$ which is the defining operator of a principal integral of Airy function. One has that $\cB(L)=1$. However $L$ is stable, i.e. $\sind(L)=0$. In fact, one has that for every $i\geq 1$,
$$
   L_i=D^3-xD+i
$$
is an $i$-th principal integral of $L$. Hence the bound in Remark~\ref{rem:stabilityindex} is not tight.
}
\end{exam}
\begin{defn}
\label{def:stability}
A nonzero operator $L\in \DR$ is said to be stable under integration, or simply stable, if \ $\sind(L)=0$. An $f\in \cK$ is said to be stable if \ $\sind(f)=0$.
\end{defn}
\begin{remark}
%Suppose that $L$ is stable and $L_i$ is an $i$-th %principal integral of $L$ for all $i\geq 0$.
$L$ is a nonzero operator and $L_i$ is an $i$-th principal integral of $L$ for all $i\geq 0$. Then Corollary~\ref{cor:criterion} implies that $\ord(L_i)=\ord(L_{\sind(L)})$ for all $i\geq \sind(L)$.
\end{remark}

%\begin{remark}
%\label{rem:strandedcase}
%    {Although one can imply $\ord(L_i) = \ord(L_{i+1})$ when $i>\sind(L)$, the converse is usually not true. Suppose that $L:=(x^9+x^5+1)D^2+(9x^8+x^7+1)D+1$. Using Abramov-van Hoeij's algorithm iteratively, one can see that $\ord(L_5)\neq\ord(L_6) = \ord(L_7)\neq\ord(L_8)$. The indicial polynomial of $L$ at $\infty$ is $x^2-8x$.}
%\end{remark}
Below are several examples of stable operators.
\begin{exam}
\label{exam:onepole}
Suppose that $\beta\in C$ is not a positive integer and $L=D+\beta/(x-c)$ for some $c\in C$. Then for each $s\geq 1$,
$$
   L^*((x-c)^s)=-s(x-c)^{s-1}+\beta(x-c)^{s-1}=(\beta-s)x^{s-1}.
$$
Since $\beta$ is not a positive integer, $\beta-s\neq 0$ for any $s\geq 1$. Hence $\deg(L^*((x-c)^s))=s-1$ for all $s\geq 1$ and then $\delta(L,(\beta-s)(x-c)^{s-1})=1$ for all $s\geq 1$. So $\sind(L)=0$, i.e. $L$ is stable.
\end{exam}

\begin{exam}
\label{exam:constantcoefficients}
Suppose that $L=p(D)$, where $p\in C[z]$ is a polynomial of positive degree. Let $s$ be the maximal integer such that $z^s$ divides $p$. Then for each $i\geq s$,
$$
   L^*(x^s)=(-1)^sc_s{i \choose s} x^{i-s}+\mbox{lower terms}
$$
where $c_s$ is the trailing coefficient of $p$. Thus $\sind(L)=0$, i.e. $L$ is stable.
\end{exam}

\begin{exam}
Suppose $\alpha,\beta,\gamma\in C$. Consider  {the} hypergeometric differential operator
$$
    L=D^2+\frac{(\alpha+\beta+1)x-\gamma}{x(x-1)}D+\frac{\alpha\beta}{x(x-1)}.
$$
For each $s\geq 0$, an easy calculation yields that
$$
L^*(x^{s+1}(x-1))=(s+1-\alpha)(s+1-\beta)x^s-s(s+1-\gamma)x^{s-1}.
$$
Suppose that $\alpha-1\notin \bZ_{\geq 0}$ and $\beta-1\notin \bZ_{\geq 0}$. Then $\deg((L^*(x^{s+1}(x-1)))=s$. Hence $L$ is stable if neither $\alpha-1$ nor $\beta-1$ is a nonnegative integer.
\end{exam}

\section{Stability Indices of special operators} \label{sec:stability}
For a general operator {,} even for a first order operator, it is difficult to compute its stability index. In this section, we shall study the stability indices of some special operators.
\begin{prop}
\label{prop:propertyofstable}
%Suppose that $L$ is stable.
 {The set $P_L=\{p\in C[x]| \delta(L,p)=1\}$ is a $C$-vector space for arbitrary nonzero operator $L$.
 Further if $L$ is stable, then $P_L=C[x]$.}
\end{prop}
\begin{proof}
 {Suppose that $p_1,p_2\in P_L$ and $a_1,a_2\in C$. Then there are $l_1,l_2\in C(x)$ such that $L^*(l_1) = p_1$ and $L^*(l_2) = p_2$. Then $L^*(a_1l_1+a_2l_2)=a_1p_1+a_2p_2$. Hence $a_1p_1+a_2p_2\in P_L$ and so $P_L$ is a $C$-vector space.}

 {Suppose that $L$ is stable. Then for each $i\geq 0$, there is a nonzero $h_i\in C[x]$ of degree $i$ such that $\delta(L,h_i)=1$. Let $r_i\in C(x)$ be a solution of $L^*(y)=h_i$. Note that $\{h_i \mid i\geq 0\}$ is a basis of $C[x]$ as a vector space over $C$. Assume that $p\in C[x]$. Write $p=\sum_{i=0}^\ell c_i h_i$ where $\ell=\deg(p)$. Then $\sum_{i=0}^\ell c_i r_i$ is a solution of $L^*(y)=p$ in $C(x)$.}
\end{proof}
%\begin{remark}\label{integralbasis}

%\end{remark}
\subsection{Stability indices of Katz's operators}
In~\cite{Katz1987}, Katz introduced a class of operators of the following form:
\begin{equation}
\label{eqn:Katz}
    p(D)+q
\end{equation}
where $p,q\in C[z]\setminus C, p(0)=0$ and $\gcd(\deg(p), \deg(q))=1$.
He proved that differential Galois groups of the operators of the above form are large. In this subsection, let us compute the stability indices of Katz's operators.
\begin{lem}
\label{lm:stabilityindex}
Suppose that $L$ is of the form
$$
  D^n+\frac{a_{n-1}}{a_n}D^{n-1}+\dots+\frac{a_0}{a_n},\,\, a_i\in C[x],\,\, a_n\neq 0.
$$
Suppose further that \  {$\gcd(a_0\ ,a_1\ ,\dots\ ,a_n)=1$} and  {$\deg(a_0)>\deg(a_i)-i$} for all $i=1,2,\dots,n$. Then $\sind(L)\leq \deg(a_0)$.
\end{lem}
\begin{proof}
For each $s\geq 0$, a simple calculation yields that $L^*(a_nx^s)$ is a polynomial of degree $\deg(a_0)+s$. Thus $\sind(L)\leq \deg(a_0)$.
\end{proof}
\begin{remark}
\label{rem:boundcompare}
 {The bound given in \ref{lm:stabilityindex} is better than the bound given in Theorem~\ref{thm:bound}. In fact, one has that $\ind^{L^*}(s)$ is a nonzero constant and $\sigma^{L^*}=\deg(a_0)$. Thus $\cB(L)>\deg(a_0)$. }
\end{remark}
\begin{prop}
\label{prop:katz}
Suppose that $L=p(D)+q(x)$, where $p,q\in C[z]$ and $p(0)=0$. Then $\sind(L)=\max\{\deg(q),0\}$.
\end{prop}
\begin{proof}
If $q(x)=0$ then from Example~\ref{exam:constantcoefficients} one sees that
$$
  \sind(L)=0=\max\{\deg(q),0\}.
$$
Assume that $q(x)\neq 0$.
By Lemma~\ref{lm:stabilityindex}, $\sind(L)\leq \deg(q)$. It remains to show that $\delta(L,d)=0$ for any polynomial $d$ of degree $\deg(q)-1$. Suppose on the contrary that there is an $l\in k$ such that $L^*(l)$ is a polynomial of degree $\deg(q)-1$. Suppose that $l$ is not a polynomial and $c\in C$ is a pole of $l$. Then $c$ must be a pole of $L^*(l)$ and so $L^*(l)$ can not be a polynomial. However, this implies that $\deg(L^*(l))=\deg(q)+\deg(l)>\deg(q)-1$, a contradiction. Hence $\delta(L,d)=0$ for any polynomial $d$ of degree $\deg(q)-1$.
\end{proof}

\subsection{Stability indices of first-order operators}
\label{sec:firstorderoperators}
In this subsection, we shall focus on studying the stability indices of first-order operators.
Suppose that $f\in k$ and $c\in C$. Let
$$
   f=\sum_{i\geq \ell} a_i (x-c)^i \quad  \text{with $a_i\in C$ and $a_\ell\neq 0$}
$$
be the power series expansion of $f$ at $c$. Then $\ell$ is called the order of $f$ at $c$, denoted by $\ord_c(f)$, and $a_{-1}$ is called the residue of $f$ at $c$, denoted by $\res_c(f)$. Similarly, we can define the order and residue of $f$ at $\infty$. Let
$$
    f=\sum_{i\geq \ell} a_i \left(\frac{1}{x}\right)^{i}, a_i\in C, a_\ell\neq 0
$$
be the power series expansion of $f$ at $\infty$. Then $\ell$ is called the order of $f$ at $\infty$, denoted by $\ord_\infty(f)$, and $a_1$ (not $a_{-1}$ in this case) is called the residue of $f$ at $\infty$, denoted by $\res_\infty(f)$. Later one will see that the residues of $f$ play an important role in estimating the stability index of $D+f$. We shall use $\cS(f)$ to denote the set of $c\in C$ such that $c$ is a simple pole of $f$ and $\res_c(f)$ is a negative integer. As usual, we use $\den(f)$ to denote the denominator of $f$.
\begin{notation}
\label{not:Delta}
For a rational function $f \in k$, set
$$
    \Delta(f)=%\begin{cases}
              %1 & \cS(f)=\emptyset\\
             { \prod\limits_{c\in \cS}(x-c)^{-\res_c(f)} }
            %&\cS(f)\neq \emptyset.
   %\end{cases}.
$$
\end{notation}
\begin{lem}
\label{lm:denominators}
Suppose that $L=D+f$ and $r\in k$. If $L^*(r)$ is a nonzero polynomial then the denominator of $r$ divides $\Delta(f)$.  {Furthermore, the zeros of the denominator of $r$ has the same multiplicity as $\Delta(f)$.}
\end{lem}
\begin{proof}
Set $r_2=\den(r)$. If $r_2\in C$ then there is nothing to prove. Suppose that $r_2\notin C$. Let $c\in C$ be a zero of $r_2$ with multiplicity $m$. It suffices to show that $(x-c)^m$ divides $\Delta(f)$. One has that $\ord_c(r)=-m$ and $\ord_c(r')=-m-1$. Write $r=r_0 (x-c)^{-m}+\cdots$ and  $f=f_0 (x-c)^\mu+\cdots$, where $r_0, f_0\in C$ and $\mu=\ord_c(f)$. From $-r'+fr=d$, one sees that $-m-1=-m+\mu$ and $mr_0+f_0r_0=0$. These imply that $\mu=-1$ and $f_0=-m$. Consequently, $c$ is a simple pole of $f$ and $\res_c(f)=f_0=-m$. So $(x-c)^m$ divides $\Delta(f)$  {and the multiplicity of $c$ in $\Delta(f)$ is equal to $m$.}
\end{proof}

The following result reduces the stability problem on first order operators
into that on operators of special form.
\begin{prop}
\label{prop:reduction1}
Suppose that $L=D+f$ and $h\in C[x]$ is a non zero polynomial. If $L$ is stable then $L-\frac{h'}{h}$ is stable.
\end{prop}
\begin{proof}
For each $i\geq 0$, let $r_i$ be a solution of $L^*(y)=x^ih$ in $k$. Such $r_i$ exists because of Proposition~\ref{prop:propertyofstable}.
A simple calculation yields that
$$
   \left(L-\frac{h'}{h}\right)^*\left(\frac{r_i}{h}\right)=-\left(\frac{r_i}{h}\right)'+\frac{fr_i}{h}-\frac{r_ih'}{h}=\frac{-r_i'+fr_i}{h}=x^i
$$
i.e. $\delta(L-h'/h,x^i)=1$ for all $i\geq 0$. So $L-h'/h$ is stable.
\end{proof}

The converse of Proposition~\ref{prop:reduction1}  {sometimes} is not true. For example, $D$ is stable but $D+1/x$ is not (see Corollary~\ref{cor:stabilityindexfirstorder} for a proof). However, if $h$ is a special divisor of the denominator of $f$ then the converse is still true.
\begin{lem}
\label{lm:solutionforrational}
Suppose that $L=D+f$ and $c\in C$ is a pole of $f$.
\begin{enumerate}
    \item For each $i\geq \sind(L)+1$, there is a polynomial $p_i\in C[x]$ of degree $i$ such that $L^*(y)=p_i/(x-c)$ has a solution in $k$.
   \item If $L$ is stable then for each $i\geq 0$, $L^*(y)=x^i/(x-c)$ has a solution in $k$.
\end{enumerate}{}
\end{lem}
\begin{proof}
(1). Write $\den(f)=(x-c)M$. One has that
$$
  L^*(M)=-M'+fM=-M'+f_1/(x-c)=\tilde{M}+\alpha/(x-c)
$$
where $f_1$ is the numerator of $f$, $\tilde{M}\in C[x]$ and $\alpha\in C$. For each $j\geq \sind(L)$, let $d_j\in C[x]$ be a polynomial of degree $j$ such that $L^*(y)=d_j$ has a solution $r_j\in k$. Fix an $i\geq \sind(L)+1$. It is easy to see that there are $\ell>0, \beta_1,\dots,\beta_\ell\in C$ such that $h=\tilde{M}+\sum_{j=\sind(L)}^\ell {\beta}_j d_j$ is of degree $i-1$. Then
$$
   L^*\left(M+\sum_{j=\sind(L)}^\ell \beta_j r_j\right)=h+\frac{\alpha}{x-c}=\frac{(x-c)h+\alpha}{x-c}.
$$
Since $\deg((x-c)h+\alpha)=i$, $(x-c)h+\alpha$ is a polynomial as required.

(2). As in (1), one has that $L^*(M)=\tilde{M}+\alpha/(x-c)$. As $\gcd(f_1,x-c)=1$, $\alpha\neq 0$. If $\tilde{M}=0$ then $M/\alpha$ is a solution of $L^*(y)=1/(x-c)$. Otherwise, let $r$ be a solution of $L^*(y)=\tilde{M}$. Such $r$ exists due to Proposition~\ref{prop:propertyofstable}. Then $(M-r)/\alpha$ is a solution of $L^*(y)=1/(x-c)$. This proves the case $i=0$. For the case $i>0$, write $x^i/(x-c)=p+\beta/(x-c)$ where $p\in C[x]$ and $\beta\in C$. Let $r_1$ be a solution of $L^*(y)=p$ in $k$ and $r_2$ a solution of $L^*(y)=\beta/(x-c)$. Then $r_1+r_2$ is a solution of $L^*(y)=x^i/(x-c)$ in $k$.
\end{proof}
\begin{prop}
\label{prop:reduction2}
Suppose that $L=D+f$ and $c\in C$ is a pole of $f$. Then
\begin{enumerate}
    \item $\sind(L+1/(x-c))\leq \sind(L)+1$;
    \item if \ $L$ is stable then $L+1/(x-c)$ is also stable.
\end{enumerate}
\end{prop}
\begin{proof}
(1). For each $i\geq \sind(L)+1$, by Lemma~\ref{lm:solutionforrational}, there exists a polynomial $p_i$ of degree $i$ such that $L^*(y)= {\frac{p_i}{x-c}}$ has a solution $r_i$ in $k$. Then
$$
  \left(L+\frac{1}{x-c}\right)^*((x-c) {r_i})=(x-c)(-r_i'+fr_i)=p_i.
$$
Hence $\sind(L+1/(x-c))\leq \sind(L)+1$.

(2). Use a similar argument as in (1).
\end{proof}

We now apply Theorem~\ref{thm:bound} to the first order operators to obtain an upper bound on the stability indices of such
operators. We use $\num(f)$ to denote the numerator of a rational function $f$.
\begin{prop}
\label{prop:upperbound}
Suppose that $L=D+f$ and $\nu=\res_\infty(f)$. Then
\[
\sind(L)\leq \begin{cases}
   \max\{\deg(f_1),\deg(f_2)\} & \ord_\infty(f)\neq 1\\
    \deg(f_2)-1 &  \ord_\infty(f)=1\,\,\mbox{and}\,\,\nu\notin \bZ_{\leq 0}\\
   -\nu+\deg(f_2) & \mbox{otherwise}
\end{cases}
\]
where $f_1=\num(f)$ and $f_2=\den(f)$.
\end{prop}
\begin{proof}
It suffices to show that $\cB(L)$ equals the corresponding bounds.
Let $f=\alpha \left(\frac{1}{x}\right)^t+\dots$ be the power series expansion of $f$ at $\infty$, where $t=\ord_\infty(f)$. Then for each $s\geq 0$,
$$
   L^*(x^s)=-sx^{s-1}+\alpha x^{s-t}+\dots.
$$
Suppose that $t>1$. Then $\ind^{L^*}(s)=-s$ and $\sigma^{L^*}=-1$. Thus $\cB(L)=\deg(f_2)=\max\{\deg(f_1),\deg(f_2)\}$ because $\deg(f_1)<\deg(f_2)$. Suppose that $t<1$. Then $\ind^{L^*}(s)=\alpha$ and $\sigma^{L^*}=-t$. Hence $\cB(L)=-t+\deg(f_2)=\deg(f_1)=\max\{\deg(f_1),\deg(f_2)\}$. Now suppose that $t=1$. Then $\alpha=-\res_\infty(f)=-\nu$. Furthermore, $\ind^{L^*}(s)=-(s-\alpha)$ and $\sigma^{L^*}=-1$. If $\nu$ is not a nonnegative integer then $\max V_{\bZ_{\geq 0}}(\ind^{L^*})=-1$ and so $\cB(L)=\deg(f_2)-1$. Otherwise if $\nu$ is a nonnegative integer then $\max V_{\bZ_{\geq 0}}(\ind^{L^*})=-\nu$ and thus $\cB(L)=-\nu+\deg(f_2)$.
\end{proof}
In what  {follows}, we shall present a lower bound.
\begin{lem}
\label{lm:lowerbound}
Suppose that $L=D+f$ with $\cS(f)=\emptyset$. Set $N=\deg(\den(f))$ and $\nu=\res_\infty(f)$.
 Then
\[
\sind(L)=\begin{cases}
   N-\min\{1,\ord_\infty(f)\} & \ord_\infty(f)\neq 1\,\,\mbox{or}\,\,\nu\notin \bZ_{\leq -N}\\
   -\nu & \mbox{otherwise}
\end{cases}.
\]
\end{lem}
\begin{proof}
 Suppose that $r\in k$ satisfies that $L^*(r)=d$ for some nonzero polynomial $d$. By Corollary~\ref{cor:criterion}, $r$ must be a polynomial. Furthermore, from
$
    -r'+rf=d,
$
one sees that $\den(f)$ divides $r$. In particular, $\deg(r)\geq \deg(\den(f))=N$.
Let
\begin{align*}
   r=c x^{s}+\dots ,\,\,
   f=\alpha x^{-t}+\dots,
\end{align*}
be the power series expansions of $r$ and $f$ respectively,
where $s=\deg(r)$ and $t=\ord_\infty(f)$. From $-r'+rf=d$ again, one has that
\begin{equation}
\label{eqn:expansion}
   -scx^{s-1}+\dots+c\alpha x^{s-t}+\dots=d.
\end{equation}

Suppose that $t=\ord_\infty(f)<1$. Then
$
  \deg(d)
$
must be equal to $s-t$ that is not less than $N-t$.
In other words, if $\delta(L,d)=1$ for some nonzero polynomial $d$ then $\deg(d)\geq N-t$. So $\sind(L)\geq N-t$. On the other hand, for each $i\geq 0$,
$$
 L^*(\den(f)x^i)=-\alpha x^{N+i-t}+\mbox{lower terms}.
$$
 This implies that for each $s\geq N-t$, $\den(f)x^{s-N+t}$ is a solution of $L^*(y)=-\alpha x^s+\mbox{lower terms}$, i.e. $\sind(L)\leq N-t$. Consequently, $\sind(L)=N-t$.
Similarly, one can prove the case $t>1$.

Now assume that $t=1$.
In this case, $\alpha=-\res_\infty(f)=-\nu$ and $s=\deg(r)\geq N>0$. Moreover, the equality (\ref{eqn:expansion}) becomes
\begin{equation}
\label{eqn:orderone}
-c(s+\nu)x^{s-1}+\mbox{lower terms}=d.
\end{equation}
Suppose that $\nu\notin \bZ_{\leq -N}$. Then $s+\nu\neq 0$ for all $s\geq N$. Hence $\deg(d)$ must be equal to $s-1$. Using a similar argument as in the case $t<1$, one has that $\sind(L)=N-1$.
It remains to show the case $\nu\in \bZ_{\leq -N}$. We first show that $\delta(L,d)=0$ for any polynomial $d\in C[x]$ of degree $-\nu-1$. Suppose on the contrary that there is a polynomial $d\in C[x]$ of degree $-\nu-1$ such that $L^*(y)=d$ has a solution $r\in C[x]$. If $s=\deg(r)>-\nu$ then $s+\nu=0$ by (\ref{eqn:orderone}). This implies that $s=-\nu$, a contradiction. So $s\leq -\nu$. From (\ref{eqn:orderone}), it is obvious that $s$ can not be less than $-\nu$. Thus $s=-\nu$. While, this means that the degree of the left-hand side of (\ref{eqn:orderone}) is less than $-\nu-1$. This contradicts  {the} assumption that $\deg(d)=-\nu-1$. Consequently, $\delta(L,d)=0$ for any polynomial $d$ of degree $-\nu-1$ and so $\sind(L)\geq -\nu$. Finally, for each $i\geq 0$,
$$
   L^*\left(\den(f)x^{i+1-\nu-N}\right)=-(i+1)x^{i-\nu}+\mbox{lower terms}.
$$
This implies that $\sind(L)\leq -\nu$. Therefore $\sind(L)=-\nu$.
\end{proof}
\begin{prop}
\label{prop:lowerbound}
Suppose that $L=D+f$. Set
$$
    N=\deg(\den(f))-|\cS(f)|-\deg(\Delta(f))
$$
and $\nu=\res_\infty(f).$ Then
\[
\sind(L)\geq \begin{cases}
  N-\min\{1,\ord_\infty(f)\} & \ord_\infty(f)\neq 1\,\,\mbox{or}\,\,\nu\notin \bZ_{\leq -N}\\
   -\nu & \mbox{otherwise}
\end{cases}.
\]
\end{prop}
\begin{proof}
Set $M_1=\prod_{c\in \cS(f)}(x-c)$ and write $\den(f)=M_1M_2$. Then $\deg(M_2)=\deg(\den(f))-|\cS(f)|$. Since $\gcd(M_1,M_2)=1$, there exist $q_1,q_2\in C[x]$ such that
$
   f=q_1/M_1+q_2/M_2.
$
Since $M_1$ is square-free and $\res_c(q_1/M_1)=\res_c(f)$ for each $c\in \cS(f)$, one has that
$$
 \frac{q_1}{M_1}=\sum_{c\in \cS(f)} \frac{\res_c(f)}{x-c}=-\frac{\Delta(f)'}{\Delta(f)}.
$$
So $f=-\Delta(f)'/\Delta(f)+q_2/M_2$. Set $\tilde{L}=L+\Delta(f)'/\Delta(f)=D+q_2/M_2$. Using Proposition~\ref{prop:reduction2} repeatedly yields that $\sind(\tilde{L})\leq \sind(L)+\deg(\Delta(f))$. Since $\cS(q_2/M_2)=\emptyset$, by Lemma~\ref{lm:lowerbound}, one has that
if $\ord_{\infty}(q_2/M_2)\neq 1$ \mbox{or} $\res_\infty(q_2/M_2)\notin \bZ_{\leq -\deg(M_2)}$ then $\sind(\tilde{L})=\deg(M_2)-\min\{1,\ord_\infty(q_2/M_2)\}$, otherwise $\sind(\tilde{L})=-\res_\infty(q_2/M_2)$.

Note that $\ord_\infty(-\Delta(f)'/\Delta(f))=1$ and $\res_\infty(-\Delta(f)'/\Delta(f))=\deg(\Delta(f))$. Furthermore, one has that
\begin{align*}
   \res_\infty(f)&\geq \min\left\{\ord_\infty\left(-\frac{\Delta(f)'}{\Delta(f)}\right),\ord_\infty\left(\frac{q_2}{M_2}\right)\right\}\\
   &=\min\left\{1,\ord_\infty\left(\frac{q_2}{M_2}\right)\right\}
\end{align*}
and the equality holds if $\ord_\infty(q_2/M_2)\neq 1$.
If $\ord_\infty(f)<1$ then $\ord_\infty(q_2/M_2)=\ord_\infty(f)<1$ and so
\begin{align*}
    \sind(L)& \geq \sind(\tilde{L})-\deg(\Delta(f))\\
    &=\deg(M_2)-\min\{1,\ord_\infty(q_2/M_2)\}-\deg(\Delta(f))\\
    &=N-\min\{1,\ord_\infty(f)\}.
\end{align*}
Suppose that $\ord_\infty(f)>1$. Then $\ord_\infty(q_2/M_2)=1$ and moreover $\res_\infty(q_2/M_2)=-\res_\infty(-\Delta(f)'/\Delta(f))=-\deg(\Delta(f))$. Assume that $\deg(\Delta(f))\geq \deg(M_2)$. Then $\res_\infty(q_2/M_2)\in \bZ_{\leq -\deg(M_2)}$ and so
$\sind(\tilde{L})=-\res_\infty(q_2/M_2)=\deg(\Delta(f)).$
Thus $\sind(L)\geq 0$. On the other hand, in this case,
$$
  N-\min\{1,\ord_\infty(f)\}=\deg(M_2)-\deg(\Delta(f))-1<0\leq \sind(L).
$$
Assume that $\deg(\Delta(f))<\deg(M_2)$. Then we have $\res_\infty(q_2/M_2)\notin \bZ_{\leq -\deg(M_2)}$. So
$\sind(\tilde{L})=\deg(M_2)-1$ and $\sind(L)\geq \sind(\tilde{L})-\deg(\Delta(f))=N-1$. Since $\ord_\infty(f)>1$, we have $N-1= N-\min\{1,\ord_\infty(f)\}$ and then the assertion holds.

Note that $\nu=\res_\infty(f)=\res_\infty(-\Delta(f)'/\Delta(f))+\res_\infty(q_2/M_2)$ and $\res_\infty(-\Delta(f)'/\Delta(f))=\deg(\Delta(f))$. This implies that $\nu\notin \bZ_{\leq -N}$ if and only if $\res_\infty(q_2/M_2)\notin \bZ_{\leq -\deg(M_2)}$. Suppose that $\ord_\infty(f)=1$ and $\nu\notin \bZ_{\leq -N}$. Then $\ord_\infty(q_2/M_2)\geq 1$ and $\sind(\tilde{L})=\deg(M_2)-1$. Thus $\sind(L)\geq \sind(\tilde{L})-\deg(\Delta(f))\geq N-1$. The assertion then follows from the fact that $\min\{1,\ord_\infty(f)\}=1$. Finally, assume that $\ord_\infty(f)=1$ and $\nu\in \bZ_{\leq -N}$. Then $\res_\infty(q_2/M_2)\in \bZ_{\leq -\deg(M_2)}$. Hence $\sind(\tilde{L})=-\res_\infty(q_2/M_2)$ and then $\sind(L)\geq -\res_\infty(q_2/M_2)-\deg(\Delta(f))=-\nu$.
\end{proof}
% \begin{exam}{}
% \label{exam:stabilityindex}
% Suppose that
% $$
%   f=\frac{n}{x}-\frac{m}{x-1}
% $$
% where $m,n$ are positive integers.
% Then $\cS(f)=\{1\}$ and $\Delta(f)=(x-1)^m$. Furthermore $\ord_\infty(f)\geq 2$. Hence by Proposition~\ref{prop:lowerbound}, $\sind(D+f)\geq \sum_{i=1}^n \ell_i-1$. In fact, we have that $\sind(D+f)=\sum_{i=1}^n \ell_i-1$. To see this, for each $s\geq \sum_{i=1}^n \ell_i-1$, write $s+1=\sum_{i=1}^n m_i$ with $m_i\geq \ell_i$ for all $1\leq i \leq n$. Set $h=\prod_{i=1}^n (x-c_i)^{m_i}$. Then
% $$
%    (D+f)^*(h)=-(s+1)x^s+\mbox{lower terms}.
% $$
% In other words, $\delta(L,d_s)=1$ for some polynomial $d_s\in C[x]$ of degree $s$. So $\sind(D+f)\leq \sum_{i=1}^n \ell_i-1$ and thus the equality holds.
% \end{exam}
\section{Stable first order operators}\label{sec:hyperexp}
Any solution of a first order operator over $k$ is called a hyperexponential function over $k$.
A first order operator is stable if and only if its corresponding hyperexponential solution is stable.
The stability problem  {for} hyperexponential functions of the form $f\exp(g)$ with $f, g\in k$ has been
solved in~\cite{chen2022ISSAC}.
In this section, we shall solve the stability problem  {for} general hyperexponential functions by presenting
a necessary and sufficient condition on the stability of first order operators.
\begin{thm}
\label{thm:criterion}
Suppose that $L=D+f$ where $f\in C(x)$. Then $L$ is stable if and only if $f$ admits one of the following forms:
$$
    -\frac{h'}{h}+\alpha \,\,\mbox{or}\,\,-\frac{h'}{h}+\frac{\beta}{x-c}, \,\,h\in C[x]\setminus\{0\}, \alpha,\beta,c\in C
$$
and $\beta$ is not a positive integer.
\end{thm}
\begin{proof}
Suppose that $f$ admits one of the above forms.
By Proposition~\ref{prop:reduction1}, it suffices to show that $D+\alpha$ and $D+\beta/(x-c)$ are stable. The case $D+\alpha$ follows from Proposition~\ref{prop:katz} and the case $D+\beta/(x-c)$ follows from Example~\ref{exam:onepole}.

Now suppose that $L$ is stable.
Set $h=\Delta(f)$ and $\tilde{f}=f+h'/h$. If $\tilde{f}=0$ then there is nothing to prove. Suppose that $\tilde{f}\neq 0$. Then there is no simple pole of $\tilde{f}$ at which the residue of $\tilde{f}$ is a negative integer, i.e. $\cS(\tilde{f})=\emptyset$ and $\Delta(\tilde{f})=1$. Using Proposition~\ref{prop:reduction2} repeatedly, one has that $\tilde{L}=L+h'/h$ is stable. Note that $\tilde{L}=D+\tilde{f}$.

 Set $\tilde{f}_1=\num(\tilde{f}_1)$ and $\tilde{f}_2=\den(\tilde{f})$.  {Note that $\cS(\tilde{f})=\emptyset$. If $\ord_\infty(\tilde{f})<0$ then $\deg(\tilde{f}_1)>0$, and by Lemma~\ref{lm:lowerbound} one has that
 $$\sind(\tilde{L})\geq \deg(\tilde{f}_2)-\ord_\infty(\tilde{f})=\deg(\tilde{f}_1)>0$$
 a contradiction.
 If $\ord_\infty(\tilde{f})>1$ then $\deg(\tilde{f}_2)>1$ and $\sind(\tilde{L})\geq \deg(\tilde{f}_2)-1>0$ by Lemma~\ref{lm:lowerbound}. One obtains a contradiction again.
 Therefore $\ord_\infty(\tilde{f})=0$ or $\ord_\infty(\tilde{f})=1$.}

  {Suppose that $\ord_\infty(\tilde{f})=0$. Then by Lemma~\ref{lm:lowerbound} again, $0=\sind(\tilde{L})=\deg(\tilde{f}_1)$. This implies that $\deg(\tilde{f}_2)=0$ for $\ord_\infty(\tilde{f})=0$.
 Hence one has that $\alpha=\tilde{f}\in C\setminus\{0\}$ and then $f=-h'/h+\alpha$. Now we consider the case $\ord_\infty(\tilde{f})=1$. Then $\deg(\tilde{f}_2)\geq 1$ and $\res_\infty(\tilde{f})\neq 0$.  If $\res_\infty(\tilde{f})\in \bZ_{\leq -\deg(\tilde{f}_2)}$ then
 by Lemma~\ref{lm:lowerbound}, $\sind(\tilde{L})=-\res_\infty(\tilde{f})\geq \deg(\tilde{f}_2)\geq 1$, a contradiction. So $\res_\infty(\tilde{f})\notin \bZ_{\leq -\deg(\tilde{f}_2)}.$
 In this case, by Lemma~\ref{lm:lowerbound} again, one sees that
 $$
  0=\sind(\tilde{L})=\deg(\tilde{f}_2)-1.
$$
 Therefore $\deg(\tilde{f}_2)=1$ and $\deg(\tilde{f}_1)=0$.
Set $\beta=\tilde{f}_1$ and write $\tilde{f}_2=x-c$ for some $c\in C$. Then $f=-h'/h+\beta/(x-c)$.}

 It remains to show that $\beta$ is not a positive integer. Assume on the contrary that $\beta$ is a positive integer. Since $\tilde{L}$ is stable, by Proposition~\ref{prop:propertyofstable}, there is an $r\in k$ such that $\tilde{L}^*(r)=(x-c)^{\beta-1}$. As $\Delta(\tilde{f})=1$, Lemma~\ref{lm:denominators} implies that $r$ is a polynomial. From $\tilde{L}^*(r)=(x-c)^{\beta-1}$, one has that
 $$
    -r'(x-c)+r\beta=(x-c)^{\beta}.
 $$
 Write $r=\alpha x^s+\dots$, where $s=\deg(r)$. Then
 $$
    -\alpha(s-\beta)x^s+\mbox{lower terms}=(x-c)^\beta.
 $$
 Comparing the degrees of both sides of the above equality, one sees that there does not exist any integer $s$ such that the above equality holds, a contradiction. Therefore $\beta$ is not a positive integer.
\end{proof}
\begin{exam}
\label{exam:radicalfunctions}
 { Let $L=D-\frac{1}{2x}-\frac{1}{2(x+1)}$. Then $L$ is the defining operator of $\sqrt{x(x+1)}$. By Theorem~\ref{thm:criterion}, $L$ is not stable. Notice that the method presented in \cite{chen2022ISSAC} is not valid for this example.}
\end{exam}
\section{Iterated integration of special rational functions} \label{sec:rationalfunctions}
In the last section, we shall consider the iterated integrals of the rational function $1/q$, where $q$ is a nonzero polynomial in $C[x]$.
Write $q=\lc(q)\prod_{i=1}^n (x-c_i)^{m_i}$. Then the first order operator with $1/q$ as a fundamental solution is of the form
\begin{equation}
\label{eqn:firstorderoperators1}
L=D+\sum_{i=1}^n \frac{m_i}{x-c_i}.
\end{equation}
\begin{lem}
\label{lm:firstorderequivalence}
Suppose that $L=D-f'/f$ where $f\in k\setminus\{0\}$. Then for a nonzero $d\in C[x]$, $\delta(L,d)=1$ if and only if $df$ is rationally integrable.
\end{lem}
\begin{proof}
For $r\in k \setminus \{0\}$, one has that
$$
   L^*(r)=-r'-\frac{f'}{f}r=-\frac{f'r+fr'}{f}=-\frac{1}{f}\left(rf\right)'.
$$
So $L^*(r)=d$ if and only if $(-fr)'=df$. In other words, $\delta(L,d)=1$ if and only if $df$ is rationally integrable.
\end{proof}
The following lemma may seem folklore, but we did not find any precise reference. So we include its full proof
for the sake of completeness.

\begin{lem}
\label{lm:notrationalintegrable}
Suppose that $q= c_0 \prod_{i=1}^n (x-c_i)^{m_i}$ with $c_i\in C$, $m_i, n>1$ and $p\in C[x]$ is a nonzero polynomial. If either $\deg(p)<n-1$ or $\deg(p)=\deg(q)-1$, then $p/q$ is not rationally integrable.
\end{lem}
\begin{proof}
Assume that $\deg(p)<n-1$ and assume on the contrary that $p/q$ is rationally integrable, i.e. there is an $r\in C(x)$ such that $r'=p/q$. Set $r_1=\num(r)$ and $r_2=\den(r)$. We claim that $q=\tilde{q}r_2\prod_{i=1}^n(x-c_i)$ for some nonzero $\tilde{q}\in C[x]$. Note that $c\in C$ is a pole of $r_1/r_2$ if and only if $c$ is a pole of $r$. Hence we may write $r_2=\prod_{i=1}^n (x-c_i)^{l_i}$ where $l_i>0$. For each $1\leq i \leq n$, $\ord_{c_i}(r')=\ord_{c_i}(r)-1=-l_i-1$ and $\ord_{c_i}(p/q)\geq -m_i$. These imply that $m_i\geq l_i+1$ for all $1\leq i \leq n$. So $q=\tilde{q}r_2\prod_{i=1}^n(x-c_i)$ and our claim holds. Now from $r'=p/q$, one has that $q(r_1'r_2-r_1r_2')=pr_2^2$. Cancelling $r_2$ from both sides yields that
\begin{equation}
\label{eqn:form1}
     \tilde{q}(r_1'r_2-r_1r_2')\prod_{i=1}^n(x-c_i)=pr_2.
\end{equation}
If $\deg(r_1)\neq \deg(r_2)$ then $\deg(r_1'r_2-r_1r_2')=\deg(r_1)+\deg(r_2)-1$. Suppose that $\deg(r_1)=\deg(r_2)$. Write $r_1=cr_2+h$ where $c\in C$ and $h\in C[x]$ with $\deg(h)<\deg(r_2)$. Since $\gcd(r_1,r_2)=1$, $h\neq 0$. One sees that $r_1'r_2-r_1r_2'=h'r_2-hr_2'$ and then $\deg(r_1'r_2-r_1r_2')=\deg(h'r_2-hr_2')=\deg(h)+\deg(r_2)-1$. In either case, one sees that $\deg(r_1'r_2-r_1r_2')\geq \deg(r_2)-1$. Then the degree of the left-hand side of (\ref{eqn:form1}) is not less than $\deg(r_2)-1+n$. However the degree of the right-hand side of (\ref{eqn:form1}) equals $\deg(p)+\deg(r_2)$ that is less than $\deg(r_2)+n-1$. We obtain a contradiction. Thus $p/q$ is not rationally integrable.

Finally suppose that $\deg(p)=\deg(q)-1$. Then the residue of $p/q$ at infinity is not zero because $\ord(x^2p/q)=-1$. It is well-known that a rational function is rationally integrable if and only if all of its residues vanish. So $p/q$ is not rationally integrable.
\end{proof}

\begin{prop}
\label{prop:orders}
Let $L$ be of the form (\ref{eqn:firstorderoperators1}) and let $L_j$ be a $j$-th principal integral of $L$. Then
\[
    \ord(L_j)=\begin{cases}
       \ord(L) + j, &0\leq j\leq n-1   \\
       \ord(L) + n - 1, &n-1<j\leq \sum_{i=1}^{n}m_i-1\\
       \ord(L) + n, &j\geq \sum_{i=1}^{n}m_i
    \end{cases}.
\]
\end{prop}
\begin{proof}
Note that $\frac{q'}{q}=\sum_{i=1}^n \frac{m_i}{x-c_i}$. Due to Lemmas~\ref{lm:firstorderequivalence} and~\ref{lm:notrationalintegrable},
$\delta(L,d)=0$ for any nonzero polynomial $d$ of degree less than $n-1$. By Corollary~\ref{cor:criterion}, $\ord(L_{j+1})=\ord(L_j)+1$ for all $0\leq j <n-1$. Thus $\ord(L_j)=\ord(L)+j$ for all $0\leq j \leq n-1$.

Set $N=\sum_{i=1}^n m_i$. Assume that $N>n$. Then there is at least one $m_i$ that is greater than 1. Without loss of generality, we assume that all of $m_1,\dots,m_s$ are greater than 1 and $m_{s+1}=\dots=m_n=1$. Suppose that $j$ is an integer satisfying that $n-1\leq j<N -1$.
Then $\sum_{i=1}^s (m_i-1)=\sum_{i=1}^n (m_i-1)\geq N-1-j$. Hence there exist integers $t, l_1,\dots,l_t$ satisfying that
$1\leq t \leq s, 0<l_i<m_i$ and $\sum_{i=1}^t l_i=N-1-j$. Set $r=1/\prod_{i=1}^t (x-c_i)^{l_i}$. Then
$$
    \left(\frac{1}{r}\right)'=\frac{h}{\prod_{i=1}^t (x-c_i)^{l_i+1}}
$$
where $h$ is a nonzero polynomial of degree $t-1$. Let $\tilde{d}=q (1/r)'$. Then $\tilde{d}$ is a nonzero polynomial because $l_i+1\leq m_i$, and
$$
  \deg(\tilde{d})=\deg(h)+\deg(q)-\sum_{i=1}^t(l_i+1)=j.
$$
As $(1/r)'=\tilde{d}/q$, Lemma~\ref{lm:firstorderequivalence} implies that $\delta(L,\tilde{d})=1$. Due to Corollary~\ref{cor:criterion} again, $\ord(L_{j+1})=\ord(L_j)$ for all $n-1\leq j < N-1$. Thus $\ord(L_j)=\ord(L_{n-1})=\ord(L)+n-1$ for all $n-1<j\leq N-1$.

By Lemma~\ref{lm:notrationalintegrable}, we have that $p/q$ is not rationally integrable if $\deg(p)=\deg(q)-1$. Therefore $\delta(L,d)=0$ for any $d\in C[x]$ of degree $N-1$. Corollary~\ref{cor:criterion} then implies that $\ord(L_N)=\ord(L_{N-1})+1=\ord(L)+n$. For each $j\geq N$, one has that
$$
   L^*(qx^{j-N+1})=(j-N+1)q x^{j-N}.
$$
In other words, $L^*(qx^{j-N+1})$ is a polynomial of degree $j$. By Corollary~\ref{cor:criterion}, $\ord(L_{j+1})=\ord(L_j)$ for all $j\geq N$. Hence $\ord(L_j)=\ord(L)+n$ for all $j\geq N$.
\end{proof}
\begin{cor}
\label{cor:stabilityindexfirstorder}
Let $L$ be of the form (\ref{eqn:firstorderoperators1}). Then $\sind(L)=\deg(q)$.
\end{cor}

\noindent\textbf{Acknowledgement.}
We thank the anonymous referees for their careful reading and their valuable suggestions.

% \section{Comments} \label{sec:comments}
%  {The corollary \ref{cor:criterion} is efficient for testing the stability but not practical for computing the stability index. Because we do not know the polynomial $p$ beforehand. For example, consider the operator in \ref{rem:strandedcase}. The adjoint is $L^*=\left( {x}^{9}+{x}^{5}+1 \right) {D}^{2}+ \left( 9\,{x}^{8}-{x}^{7}+
% 10\,{x}^{4}-1 \right) D-7\,{x}^{6}+20\,{x}^{3}+1
% $ and $L^*(\frac{-1}{5040}) = -{\frac{1}{5040}}+{\frac {{x}^{6}}{720}}-{\frac {{x}^{3}}{252}}$ is the polynomial $p$ of degree $6$ that we need. But if we use $p = x^6$ to test whether the order will rise, then we will fail since $L^*(y) = x^6$ has no rational solutions. To compute the stability index in general, one should compute the set $L^*(C(x))\bigcap C[x]$. }\textcolor{orange}{Maybe we should add an reference to the article 'Generalized Hermite Reduction, Creative Telescoping and Definite Integration of D-Finite Functions'\cite{BostanChyazkLairezSalvy2018}.}

%\noindent\textbf{Acknowledgement.}
%We thank the referees for their careful reading and valuable suggestions.

\bibliographystyle{plain}
\balance
%\bibliography{stable}

\end{document}